\documentclass[11pt]{article}

\usepackage{fullpage}
\usepackage[T1]{fontenc}
\usepackage{gitinfo}
\usepackage[usenames,dvipsnames]{color}
\usepackage{stmaryrd}

\RequirePackage[colorlinks=true]{hyperref}
\hypersetup{
  linkcolor=[rgb]{0.3,0.3,0.6},
  citecolor=[rgb]{0.2, 0.6, 0.2},
  urlcolor=[rgb]{0.6, 0.2, 0.2}
}
\usepackage{mathpazo}
\usepackage{xfrac}
\usepackage{bm}
\usepackage{todonotes}
\usepackage{setspace}
\usepackage{scrextend}

\usepackage{tikz}
\usetikzlibrary{shapes,arrows}
\usetikzlibrary{calc}
\usetikzlibrary{decorations.pathmorphing}

\usepackage{rpmacros}
\usepackage{amsthm}
\usepackage{thmtools,thm-restate}

\numberwithin{equation}{section}
\declaretheoremstyle[bodyfont=\it,qed=\qedsymbol]{noproofstyle}

\declaretheorem[numberlike=equation]{observation}

\declaretheorem[name=Observation,numbered=no]{observation*}

\declaretheorem[numberlike=equation]{theorem}
\declaretheorem[numberlike=equation,style=noproofstyle,name=Theorem]{theoremwp}
\declaretheorem[name=Theorem,numbered=no]{theorem*}

\declaretheorem[numberlike=equation]{lemma}
\declaretheorem[name=Lemma,numbered=no]{lemma*}

\declaretheorem[name=Corollary,numbered=no]{corollary*}

\declaretheorem[name=Proposition,numbered=no]{proposition*}

\declaretheorem[numberlike=equation]{claim}
\declaretheorem[name=Claim,numbered=no]{claim*}

\declaretheorem[name=Conjecture,numbered=no]{conjecture*}

\declaretheorem[name=Question,numbered=no]{question*}

\declaretheoremstyle[bodyfont=\it,qed=$\lozenge$]{defstyle} 

\declaretheorem[numberlike=equation,style=defstyle]{definition}
\declaretheorem[unnumbered,name=Definition,style=defstyle]{definition*}

\declaretheorem[unnumbered,name=Example,style=defstyle]{example*}

\declaretheorem[unnumbered,name=Notation=defstyle]{notation*}

\declaretheorem[unnumbered,name=Construction,style=defstyle]{construction*}

\declaretheorem[numberlike=equation,style=defstyle]{remark}
\declaretheorem[unnumbered,name=Remark,style=defstyle]{remark*}

\newcommand{\ckt}{\textbf{$\mathcal{C}$}}
\newcommand{\form}{\textbf{$\mathcal{F}$}}
\newcommand{\abp}{\textbf{$\mathcal{A}$}}

\newcommand{\size}{\operatorname{\mathsf{size}}}
\newcommand{\depth}{\operatorname{\mathsf{depth}}}
\newcommand{\coeff}{\operatorname{\mathsf{coeff}}}

\newcommand{\rt}{\mathsf{root}}
\renewcommand{\l}{\mathsf{left}}
\renewcommand{\r}{\mathsf{right}}

\renewcommand{\phi}{\varphi}
\renewcommand{\epsilon}{\varepsilon}

\newcommand{\abcd}{\mathsf{abecedarian}}

\def\abcncVF{\mathsf{abc \text{ - } VF}_{\text{nc}}}
\def\abcncVBP{\mathsf{abc \text{ - } VBP}_{\text{nc}}}
\def\abcncVP{\mathsf{abc \text{ - } VP}_{\text{nc}}}

\newcommand{\Det}[1]{\mathrm{Det}_{#1}}
\newcommand{\Perm}[1]{\mathrm{Perm}_{#1}}
\newcommand{\imm}[2]{\mathrm{IMM}_{#1,#2}}
\newcommand{\esym}[2]{\mathrm{ESYM}_{#1,#2}}
\newcommand{\oesym}[2]{\mathrm{ESYM}_{#1,#2}^{\mathsf{(ord)}}}
\newcommand{\chsym}[2]{\mathrm{CHSYM}_{#1,#2}}
\newcommand{\ochsym}[2]{\mathrm{CHSYM}_{#1,#2}^{\mathsf{(ord)}}}
\newcommand{\lchsym}[2]{\mathrm{\mathsf{linked}\_CHSYM}_{#1,#2}}

\usepackage{nth}
\usepackage{intcalc}
\usepackage{etoolbox}
\usepackage{xstring}

\usepackage{ifpdf}
\ifpdf
\else
\usepackage[quadpoints=false]{hypdvips}
\fi

\newcommand{\ECCCurl}[2]{http://eccc.hpi-web.de/report/\ifnumcomp{#1}{>}{93}{19}{20}#1/#2/}

\newcommand{\shortECCC}[2]{\texttt{\href{http://eccc.hpi-web.de/report/\ifnumcomp{#1}{>}{93}{19}{20}#1/#2/}{eccc:TR#1-#2}}}

\newcommand{\parseECCC}[1]{
\StrSubstitute{#1}{TR}{}[\tmpstring]%
\IfSubStr{\tmpstring}{/}{ 
\StrBefore{\tmpstring}{/}[\ecccyear]%
\StrBehind{\tmpstring}{/}[\ecccreport]%
}{
\StrBefore{\tmpstring}{-}[\ecccyear]%
\StrBehind{\tmpstring}{-}[\ecccreport]%
}%
\shortECCC{\ecccyear}{\ecccreport}}

\onehalfspace

\title{Separating ABPs and Some Structured Formulas\\ in the Non-Commutative Setting}
\author{Prerona Chatterjee \thanks{Tata Institute of Fundamental Research, Mumbai, India. Email: \texttt{prerona.ch@gmail.com}. Research supported by the Department of Atomic Energy, Government of India, under project number RTI4001.}}

\begin{document}
\maketitle


\begin{abstract}
	The motivating question for this work is a long standing open problem, posed by Nisan \cite{N91}, regarding the relative powers of algebraic branching programs (ABPs) and formulas in the non-commutative setting.
	Even though the general question continues to remain open, we make some progress towards its resolution.
	To that effect, we generalise the notion of ordered polynomials in the non-commutative setting (defined by \Hrubes, Wigderson and Yehudayoff \cite{HWY11}) to define $\abcd$ polynomials and models that naturally compute them.
		
	Our main contribution is a possible new approach towards resolving the $\ncVF$ vs $\ncVBP$ question, via lower bounds against $\abcd$ formulas. 
	In particular, we show the following.
	
	\begin{quote}
		There is an explicit $n$-variate degree $d$ $\abcd$ polynomial $f_{n,d}(\vecx)$ such that 
		\begin{itemize}
	 		\item $f_{n, d}(\vecx)$ can be computed by an $\abcd$ ABP of size $O(nd)$;
		 	\item any $\abcd$ formula computing $f_{n, \log n}(\vecx)$ must have size that is super-polynomial in $n$.
		\end{itemize}
	\end{quote}
	We also show that a  super-polynomial lower bound against $\abcd$ formulas for $f_{\log n, n}(\vecx)$ would separate the powers of formulas and ABPs in the non-commutative setting.
\end{abstract}

\newpage

\section{Introduction}
Algebraic Circuit Complexity is the study of multivariate polynomials and their classification based on how hard it is to compute them, using various computational models. 
The most well studied model is that of algebraic circuits. 
These are directed acyclic graphs that use algebraic operations like addition and multiplication over some field or ring, to compute polynomials. 
When the underlying graph is only allowed to be a tree, the model is that of algebraic formulas.

The central question in this area is whether the class $\VNP$ (algebraic analogue of the class $\NP$) is contained in the class $\VP$ (algebraic analogue of the class $\P$).
Valiant \cite{V79} had shown that the permanent polynomial is complete for $\VNP$.
Thus, the $\VP$ vs $\VNP$ question essentially boils down to asking whether the $n \times n$ permanent can be computed by a $\poly(n)$-sized algebraic circuit.

In this paper, we are interested in polynomials that come from the non-commutative polynomial ring $\F \inangle{x_1, \ldots, x_n}$, where the indeterminates do not commute with each other (that is, $xy \neq yx$ for indeterminates $x$, $y$). 
As a consequence, any monomial in a non-commutative polynomial $f \in \F \inangle{x_1, \ldots, x_n}$ is essentially a string over the alphabet $\set{x_1, \ldots, x_n}$. 
This is a natural restriction and there has been a long line of work that studies non-commutative computation beginning with\footnote{Hyafil \cite{H77} had considered non-commutative computation before this, but the main result in that paper is unfortuantely false as shown in \cite{N91}} the seminal work of Nisan \cite{N91}.

It was shown by \Hrubes, Wigderson and Yehudayoff \cite{HWY10} that the non-commutative permanent polynomial is complete for the class $\ncVNP$ (the non-commutative version of $\VNP$).
Later Arvind, Joglekar and Raja \cite{AJR16} gave a natural polynomial that is complete for the class of $n$-variate non-commutative polynomials computable by $\poly(n)$-sized circuits (denoted by $\ncVP$). 

The question of whether the classes $\ncVP$ and $\ncVNP$ are different is the central open problem in the non-commuatative setting. 
Although the general question of showing lower bounds against non-commutative circuits remains open, there has been significant progress in restricted settings \cite{LMS16,LMP19,LLS19,ST18,FLOS20}.

With respect to the general question, \Hrubes, Wigderson and Yehudayoff \cite{HWY11} showed that a sufficiently strong super-linear lower bound for the classical sum-of-squares problem implies a separation between $\ncVP$ and $\ncVNP$.
In another related work, Carmosino, Impagliazzo, Lovett and Mihajlin \cite{CILM18} showed that proving mildly super-linear lower bounds against non-commutative circuits would imply exponential lower bounds against the same model.

One motivation for studying non-commutative computation is that it is possibly easier to prove strong lower bounds in this setting as compared to the usual commutative setting.
At least intuitively, it seems harder to \emph{cancel} monomials once they have been calculated when commutativity is not allowed amongst the variables. 
For example, the $n \times n$ determinant can be computed by an $O(n^3)$ algebraic circuit, but to the best of our knowledge there is no circuit for the non-commutative determinant of size $2^{o(n)}$. 
In fact, it was shown by Arvind and Srinivasan \cite{AS18} that if the non-commutative determinant had a $\poly$-sized circuit, then $\ncVP = \ncVNP$.

Even though a super-polynomial lower bound is not known for the non-commutative determinant against circuits, Nisan \cite{N91} gave an exponential lower bound on the number of gates in any formula computing it. 
In contrast, the best lower bound known against formulas in the commutative setting is quadratic (for the elementary symmetric polynomial) \cite{N66,K85,CKSV19}. 

A point to note about the lower bound given by Nisan however, is that the proof actually works for a computational model, called Algebraic Branching Programs (or ABPs), that is believed to be more general than algebraic formulas. 
In fact, Nisan \cite{N91} gave an exact characterisation for the size of any ABP computing a non-commutative polynomial. 
As far as we are aware, any lower bound known against general non-commutative formulas uses this characterisation and hence is essentially a lower bound against non-commutative ABPs itself.

The motivating question for this work is whether there is a separation between the powers of ABPs and formulas in the non-commutative setting. 
Let us denote the class of non-commutative polynomials over $n$ variables that can be computed by $\poly(n)$-sized ABPs by $\ncVBP$. 
Similarly, let $\ncVF$ denote the class of non-commutative polynomials over $n$ variables that can be computed by $\poly(n)$-sized formulas.
The question is essentially whether $\ncVBP$ is contained in $\ncVF$.

This question had been posed by Nisan \cite{N91}, and the only work we are aware of that has made some progress with respect to this question is the one by Lagarde, Limaye and Srinivasan \cite{LLS19}.
They show that certain syntactically restricted non-commutative formulas (called Unique Parse Tree formulas) can not compute $\imm{n}{n}$ unless they have size $n^{\Omega(\log n)}$.

In this paper, we study restrictions of a different kind.
From here on, we will only be talking about non-commutative computation unless specifically mentioned otherwise.

\subsection{Abecedarian Polynomials and Models That Compute Them}

In \cite{HWY11}, {\Hrubes} {\etal} had defined the notion of \emph{ordered} polynomials. 
A homogeneous polynomial of degree $d$ is said to be ordered if the set of variables it depends on can be partitioned into $d$ buckets such that variables occuring in position $k$ only come from the $k$-th bucket.
We generalise this notion by making the bucket indices \emph{position independent}.

In English, an abecedarian word is one in which all of its letters are arranged in alphabetical order \cite{MW19}. 
Since monomials are essentially words in the non-commutative setting, we call a polynomial $\abcd$ whenever the variables in all of its monomials are \emph{arranged alphabetically}.

We call $(X_1, \ldots, X_m)$ (or $\set{X_1, \ldots, X_m}$, since the ordering is clear from context) a bucketing system of size $m$, for the set of variables $\set{x_1, \ldots, x_n}$, if it forms an \emph{ordered} partition of $\set{x_1, \ldots, x_n}$.
A non-commutative polynomial $f \in \F \inangle{x_1, \ldots, x_n}$ is said to be $\abcd$ with respect to the bucketing system $\set{X_1, \ldots, X_m}$, if 
\begin{itemize}
	\item $\set{X_1, \ldots, X_m}$ is a bucketing system for $\set{x_1, \ldots x_n}$;
	\item the variables in each of the monomials in $f$ are arranged in \emph{non-decreasing} order of their bucket indices. 
\end{itemize}
A formal definition can be found in \autoref{sec:defn_abcd_poly}.

\subsubsection*{``Getting our Hands Dirty'' with Abecedarian Polynomials}

Before moving ahead, let us take a look at an example of an $\abcd$ polynomial.
Given a commutative polynomial $f \in \F[x_1, \ldots, x_n]$, define its non-commutative analogue, $f^{(\text{nc})}$ as follows.
\begin{quote}
	$f$ and $f^{(\text{nc})}$ look essentially the same, except that variables in every monomial in $f^{(\text{nc})}$ are arranged in non-decreasing order of their indices.
\end{quote}
Then, $f^{(\text{nc})}$ is $\abcd$ with respect to the bucketing system $\setdef{X_i}{X_i = \set{x_i}}$. 

Let us also look at a possibly important polynomial that is \emph{not} $\abcd$ with respect to the bucketing system $\setdef{X_i}{X_i = \set{x_i}}$.
Consider the \emph{arc-full rank polynomial}\footnote{The hard polynomial constructed by Dvir, Malod, Perifel and Yehudayoff \cite{DMPY12} to separate the powers of formulas and ABPs in the multilinear setting.} $f$, but as a non-commutative polynomial, $f'$, in the following sense. 
\begin{quote}
	Let $\abp$ be the ABP that computes $f$ and think of that as a non-commutative ABP $\abp'$. Then, $f'$ is the polynomial computed by $\abp'$. 
\end{quote}
It is not hard to see that across different monomials in $f'$, the order in which variables are arranged is not consistent.
Thus, $f'$ is not $\abcd$ with respect to the given bucketing system.

A final point to note before we move ahaed is that a polynomial might be $\abcd$ with respect to different bucketing systems\footnote{Every polynomial $f \in \F\inangle{x_1, \ldots, x_n}$ is $\abcd$ with respect to the bucketing system $\set{X}$ for $X = \set{x_1, \ldots, x_n}$.}. 
In fact, even the sizes of the different bucketing systems might be different. 
For example, the polynomial
\[
	\oesym{n}{d} = \sum_{1 \leq i_1 < \ldots < i_d \leq n} x^{(1)}_{i_1} \cdots x^{(d)}_{i_d}
\]
is $\abcd$ with respect to the bucketing system $\set{X_k = \setdef{x^{(k)}_i}{i \in [n]}}$ which has size $d$, as well as $\set{X_i = \setdef{x^{(k)}_i}{k \in [d]}}$ which has size $n$.

\subsubsection*{Abecedarian Models of Computation}

{\Hrubes} {\etal} \cite{HWY11} had defined \emph{ordered circuits}, a model naturally suited to compute ordered polynomials.
We generalise this notion to define circuits that naturally compute $\abcd$ polynomials.
Going one step further, we also define $\abcd$ ABPs and $\abcd$ formulas.

Suppose $f$ is an $\abcd$ polynomial with respect to the bucketing system $\set{X_1, \ldots, X_m}$. 
For any $1 \leq a \leq b \leq m+1$, $f[a,b)$ is a sub-polynomial of $f$ defined as follows.
\begin{itemize}
	\item For any $a \in [m+1]$, $f[a, a)$ is the constant term in $f$.
	\item For $1 \leq a < b \leq m+1$, $f[a,b)$ contains only those monomials of $f$ in which the first variable is from bucket $X_a$ and the last variable is from any of the buckets in the set $\set{X_a, \ldots, X_{b-1}}$. \qedhere
\end{itemize}

A circuit is said to be $\abcd$ if every gate $v$ in it can be labelled by a tuple $(a,b)$ such that if $f_v$ is the polynomial computed at that gate, then $f_v = f_v[a,b)$.
We call a formula $\abcd$ if it has a similar syntactic property at every gate. 
For formal definitions, see \autoref{def:ordCkt} and \autoref{def:ordForm} respectively.
On the other hand, an ABP is said to be $\abcd$ when every vertex in it can be labelled by a bucket index such that if $f$ is the polynomial computed between vertices labelled with indices $a$ and $b$ respectively, then $f = f[a,b+1)$. 
\autoref{def:ordABP} is a formal definition.

\subsection{Our Main Results}

Our main result is a super-polynomial separation between $\abcd$ formulas and ABPs.

\begin{restatable}[Separating Abecedarian Formulas and Abecedarian ABPs]{theorem}{mainThm}\label{thm:main}
	Define 
	\[
		\lchsym{n}{d}(\vecx) = \sum_{i_0=1}^{n} \inparen{\sum_{i_0 \leq i_1 \leq \ldots \leq i_d \leq n} x_{i_0, i_1} \cdot x_{i_1,i_2} \cdots x_{i_{d-1}, i_d}}
	\]
	to be the \emph{linked} complete homogeneous polynomial over $n$-variables of degree $d$.

	Note that $\lchsym{n}{d}(\vecx)$ is $\abcd$ with respect to the bucketing system $\setdef{X_i}{i \in [n]}$ where $X_i = \setdef{x_{i,j}}{i \leq j \leq n}$. With respect to this bucketing system,  
	\begin{enumerate}
		\item $\lchsym{n}{d}(\vecx)$ has an $\abcd$ ABP of size $O(nd)$;
		
		\item any $\abcd$ formula computing $\lchsym{n/2}{\log n}(\vecx)$ has size $n^{\Omega(\log \log n)}$.
	\end{enumerate}
	That is, there is a super-polynomial separation between $\abcd$ ABPs and $\abcd$ formulas.
\end{restatable}

Our second main result shows that in certain settings, formulas computing $\abcd$ polynomials can be assumed to be $\abcd$ without loss of generality.

\begin{restatable}[Converting Formulas into Abecedarian Formulas]{theorem}{orderForm}\label{thm:orderForm}
	Let $f$ be an $\abcd$ polynomial with respect to a bucketing system of size $m$, and $\form$ be a formula of size $s$ computing $f$.
	If $m = O(\log s)$, then there is an $\abcd$ formula  $\form'$ computing $f$ of size $\poly(s)$.
\end{restatable}

\noindent In other words, an $n^{\omega(1)}$ lower bound against $\abcd$ formulas computing any polynomial that is $\abcd$ with respect to a bucketing system of size $O(\log n)$, would result in a super-polynomial lower bound against general non-commutative formulas.\\

\noindent These statements suggest a new approach towards resolving the general $\ncVF$ vs $\ncVBP$ question.

\subsubsection*{Connections to the General $\ncVF$ vs $\ncVBP$ Question}

\autoref{thm:main} gives a separation between $\abcd$ formulas and ABPs.
On the other hand, \autoref{thm:orderForm} shows that if we are given a formula that computes a polynomial that is $\abcd$ with respect to a bucketing system of \emph{small} size, then we can assume that the formula is $\abcd$ without loss of generality.
Unfortunately, the bucketing system with respect to which our \emph{hard polynomial} from \autoref{thm:main} is $\abcd$, is \emph{not small} in size.
Thus, the general question of whether $\ncVBP$ is contained in $\ncVF$ or not still remains open.\\ 

\noindent However, there are two natural questions that arise at this point.
\begin{enumerate}
	\item Can any formula computing an $\abcd$ polynomial be converted to an $\abcd$ formula without much blow-up in size, irrespective of the size of the bucketing system? 
	\item Is there a polynomial $f$ which is $\abcd$ with respect to a bucketing system that has \emph{small} size such that $f$ witnesses a separation between $\abcd$ formulas and ABPs?
\end{enumerate}

\noindent Clearly, a positive answer to either of these questions would imply that $\ncVBP \neq \ncVF$.
In particular, a super-polynomial lower bound against $\abcd$ formulas for a polynomial very similar to the one we used to show our separation would separate $\ncVBP$ and $\ncVF$.

\begin{restatable}{corollary}{FormLBfromOrdFormLB}\label{cor:FormLBfromOrdFormLB}
	Let the polynomial $\lchsym{n}{d}(\vecx)$ be as defined in \autoref{thm:main}.
	An $n^{\omega(1)}$ lower bound against $\abcd$ formulas for $\lchsym{\log n}{n}(\vecx)$ would imply a super-polynomial separation between non-commutative ABPs and formulas.
\end{restatable}

In fact our proof technique also shows that a super-polynomial lower bound against \emph{homogeneous} formulas for our hard polynomial would separate $\ncVBP$ and $\ncVF$.

\begin{restatable}{corollary}{FormLBfromHomFormLB}\label{cor:FormLBfromHomFormLB}
	Let $\lchsym{n}{d}$ be as defined in \autoref{thm:main}. 
	An $n^{\omega(1)}$ lower bound against homogeneous formulas for $\lchsym{n}{\log n}(\vecx)$ would result in a super-polynomial separation between ABPs and formulas in the non-commutative setting. 
\end{restatable}

 \subsection{Proof Overview}\label{sec:proof-overview}

We now give a proof overview of our main theorems.

\subsubsection*{Separating Abecedarian Formulas and ABPs}

Let us first consider \autoref{thm:main}.

\mainThm*

\noindent A \emph{small} $\abcd$ ABP computing $\lchsym{n}{d}(\vecx)$ is essentially the following.

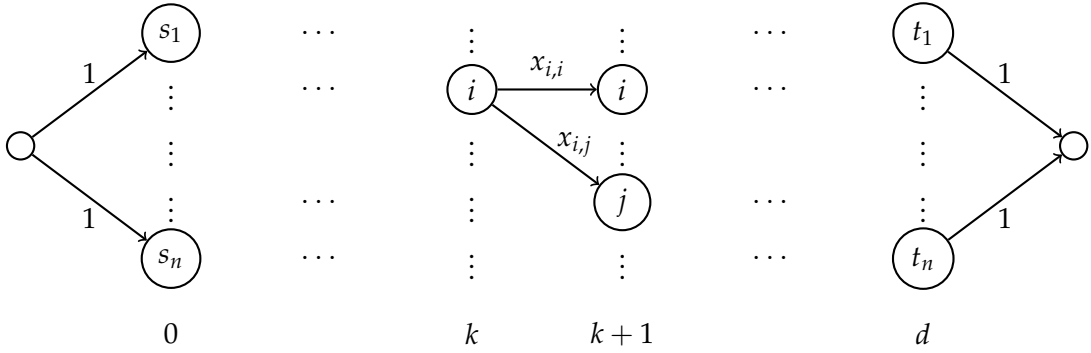
\vspace{1em}
\begin{tikzpicture}[thick, node distance=3cm]
	\node [circle, draw=black] (start) {};

	\node[circle, draw=black] (l11) at ($(start)+(2,1.5)$) {$s_1$};
	\node (l12) at ($(start)+(2,.75)$) {$\vdots$};
	\node (l10) at ($(start)+(2,0)$) {$\vdots$};
	\node (l13) at ($(start)+(2,-.75)$) {$\vdots$};
	\node[circle, draw=black] (l14) at ($(start)+(2,-1.5)$) {$s_n$};
	\node at ($(start)+(2,-2.5)$) {$0$};

	\node (l21) at ($(start)+(4,1.5)$) {$\cdots$};
	\node (l21) at ($(start)+(4,.75)$) {$\cdots$};
	\node (l23) at ($(start)+(4,-.75)$) {$\cdots$};
	\node (l24) at ($(start)+(4,-1.5)$) {$\cdots$};

	\node (l31) at ($(start)+(6,1.5)$) {$\vdots$};
	\node[circle, draw=black] (l32) at ($(start)+(6,.75)$) {$i$};
	\node (l30) at ($(start)+(6,0)$) {$\vdots$};
	\node (l33) at ($(start)+(6,-.75)$) {$\vdots$};
	\node (l34) at ($(start)+(6,-1.5)$) {$\vdots$};
	\node at ($(start)+(6,-2.5)$) {$k$};

	\node (l41) at ($(start)+(8,1.5)$) {$\vdots$};
	\node[circle, draw=black] (l42) at ($(start)+(8,.75)$) {$i$};
	\node (l40) at ($(start)+(8,0)$) {$\vdots$};
	\node[circle, draw=black] (l43) at ($(start)+(8,-.75)$) {$j$};
	\node (l44) at ($(start)+(8,-1.5)$) {$\vdots$};
	\node at ($(start)+(8,-2.5)$) {$k+1$};

	\node (l51) at ($(start)+(10,1.5)$) {$\cdots$};
	\node (l52) at ($(start)+(10,.75)$) {$\cdots$};
	\node (l53) at ($(start)+(10,-.75)$) {$\cdots$};
	\node (l54) at ($(start)+(10,-1.5)$) {$\cdots$};

	\node[circle, draw=black] (l61) at ($(start)+(12,1.5)$) {$t_1$};
	\node (l62) at ($(start)+(12,.75)$) {$\vdots$};
	\node (l60) at ($(start)+(12,0)$) {$\vdots$};
	\node (l63) at ($(start)+(12,-.75)$) {$\vdots$};
	\node[circle, draw=black] (l64) at ($(start)+(12,-1.5)$) {$t_n$};
	\node at ($(start)+(12,-2.5)$) {$d$};

	\node[circle, draw=black] (end) at ($(start)+(14,0)$) {};

	\draw[->](l32) -- node[above] {$x_{i,i}$} (l42);
	\draw[->](l32) -- node[right] {$x_{i,j}$} (l43);

	\draw[->](start) -- node[above] {$1$} (l11);
	\draw[->](start) -- node[below] {$1$} (l14);

	\draw[->](l61) -- node[above] {$1$} (end);
	\draw[->](l64) -- node[below] {$1$} (end);
\end{tikzpicture}

\vspace{1em}
\noindent For the lower bound, assume that we have been given a \emph{small} $\abcd$ formula computing the polynomial.
We then keep modifying this formula till we get a \emph{small} homogeneous multilinear formula computing the \emph{elementary symmetric polynomial} of degree $n/2$. 
We then use the known lower bound against homogeneous multilinear formulas for this polynomial (shown by {\Hrubes} and Yehudayoff \cite{HY11}), to get a contradiction.

Let us spell out the proof in some more detail.

\begin{description}
	\item[Step 1:] Assume that we are given an $\abcd$ formula  computing $\lchsym{n/2}{\log n}(\vecx)$ of size $O(n^{\epsilon \log \log n})$.
	Since the degree of the polynomial being computed is \emph{small}, we can assume that there is in fact a \emph{homogeneous} $\abcd$ formula computing $\lchsym{n/2}{\log n}(\vecx)$ of size $O(n^{c \cdot \epsilon \log \log n})$ for some constant $c$ independent of $\epsilon$.
	\item[Step 2:] Using the homogeneous $\abcd$ formula from Step 1, we obtain a \emph{structured} homogeneous $\abcd$ formula, of size $O(n^{c \cdot \epsilon \log \log n})$, computing $\lchsym{n/2}{\log n}(\vecx)$.
	\item[Step 3:] We consider the complete homogeneous polynomial over $n$ variables of degeree $d$ 
	\[
		\chsym{n}{d} (\vecx) = \sum_{1 \leq i_1 \leq \ldots \leq i_d \leq n} x_{i_1} \cdots x_{i_d},
	\]
	and show that there is a homogeneous $\abcd$ formula of size $\poly(n)$ that computes $\chsym{n/2}{\log n} (\vecx)$.
	\item[Step 4:] If the formula in Step 2 has size $s$ and that in Step 3 has size $s'$, then we show that there is a homogeneous $\abcd$ formula of size $(s \cdot s')$ computing $\chsym{n/2}{\log^2 n} (\vecx)$.
	\item[Step 5:] Next, we show that Step 4 can be used repeatedly at most $O(\sfrac{\log n}{\log \log n})$ times, to obtain a homogeneous $\abcd$ formula computing $\chsym{n/2}{n/2} (\vecx)$ of size $O(n^{c \cdot \epsilon \log n})$.
	\item[Step 6:] Using the formula obtained in Step 5, we get a homogeneous multilinear formula computing the elementary symmetric polynomial of degree $n/2$, of size $O(n^{c \cdot \epsilon \log n})$.
	\item[Step 7:] Finally, we choose $\epsilon$ in such a way that Step 6 contradicts the theorem in \cite{HY11}. 
\end{description}

\noindent The crucial observation that makes this proof work, is that the polynomial we are working with is structured enough for us to be able to amplify its degree in a systematic way (without blowing up the size by much).
This is the $4^{\text{th}}$ step in the description above.

Apart from that, the entire proof essentially boils down to the fact that when formulas are computing low degree polynomials, there are some additional tricks available to make them more structured.
A complete proof of \autoref{thm:main} can be found in \autoref{sec:main}.

We now elaborate a little on the first step, since the observations made to prove this step are quite general and possibly useful in various settings.
These statements are known to be true in the commutative setting and their proofs in the non-commutative setting are fairly similar to the ones for their commutative counterparts.
We state them here nevertheless, since to the best of our knowledge, they have not been stated formally before for the non-commutative setting.

\paragraph*{Homogenising Non-Commutative Formulas computing Low Degree Polynomials}

Raz \cite{R13} had shown that if there is a formula computing a homogeneous polynomial of \emph{low} degree in the commutative world, then it can be assumed without loss of generality that the formula is homogeneous.
We show that this statement is true even in the non-commutative setting. 

\begin{restatable}[Homogenising Non-Commutative Formulas computing Low Degree Polynomials]{lemma}{homogenisation}\label{lem:homogenisation}
	Suppose $f$ is a non-commutative homogeneous polynomial of degree $d = O(\log s)$, where $s$ is its ABP complexity. 
	Also, suppose $\form$ is a fan-in $2$ formula of size $s'$ that computes $f$.
	Then there is a homogeneous formula $\form'$ computing $f$, that has size $\poly(s')$ and whose multiplication gates have fan-in $2$.

	Further, if $\form$ was $\abcd$ with respect to some bucketing system, then $\form'$ is also $\abcd$ with respect to the same bucketing system.
\end{restatable}

The only thing that needs to be checked for Raz's proof to work in this setting is whether non-commutative formulas can be depth-reduced to log-depth.
We show that infact they can be.

\paragraph*{Depth Reduction for Non-Commutative Formulas}

Brent \cite{B74} had shown that if there is a formula of size $s$ computing a commutative polynomial $f$, then there is a formula of depth $O(\log s)$ and size $\poly(s)$ that computes the same polynomial. 
We show that this is also true in the non-commutative setting.
The proof is almost exactly along the same lines as the one by Brent \cite{B74}, just analysed carefully.

\begin{restatable}[Depth Reduction of Non-Commutative Formulas]{lemma}{depthReduction}\label{lem:depth-reduction}
	If there is a fan-in $2$ formula $\form$ of size $s$ computing a non-commutative polynomial $f$, then there is a fan-in $2$ formula $\form'$ of size $\poly(s)$ and depth $O(\log(s))$ computing $f$.

	Further if $\form$ is homogeneous, $\form'$ is also homogeneous. Similarly, if $\form$ is $\abcd$ with respect to some bucketing system, then $\form'$ is also $\abcd$ with respect to the same bucketing system.
\end{restatable}

The statement is slightly surprising because Nisan \cite{N91} had shown that such a statement is false for non-commutative circuits\footnote{Unlike commutative circuits \cite{VSBR83}.}.
\autoref{lem:depth-reduction} says that this is not the case with formulas.

This observation is not only crucial for proving some of our statements, it also answers a question posed by Nisan \cite{N91}.
For a polynomial $f$, let $D(f)$ denote the minimum depth among all fan-in $2$ circuits computing $f$, $B(f)$ the size of the smallest ABP computing $f$, and $F(f)$ the size of the smallest fan-in $2$ formula computing $f$. 
Nisan \cite{N91} had shown that if $f$ has degree $d$, then
\[
	D(f) \leq O(\log B(f) \cdot \log d) \quad \text{ and } \quad F(f) \leq 2^{D(f)}.
\]
He had then asked whether $D(f) \leq O(\log F(f))$ or not. 
\autoref{lem:depth-reduction} clearly implies that, in particular, the answer to this question is YES.

\subsubsection*{Converting Formulas into Abecedarian Formulas}

Next we go over the proof idea of \autoref{thm:orderForm}.

\orderForm*

We prove this statement by first converting the given formula $\form$ into an $\abcd$ circuit $\ckt$, and then unravelling it to get an $\abcd$ formula $\form'$ computing the same polynomial.

The first step is fairly straightforward. 
The proof is along the same lines as that for homogenising circuits.
The only difference is that we keep track of bucket indices of the variables on either ends of the monomials being computed at the gates, instead of their degrees. 

In the second step, we convert $\ckt$ into a formula $\form'$.
In order to do that, we need to recompute vertices every time it is reused.
Thus, to give an upper bound on the size of $\form'$, we need to find an upper bound on the number of distinct paths from any vertex in $\ckt$ to the root.
This analysis is done in a way similar to the one by Raz \cite{R13} in his proof of the fact that formulas computing low degree polynomials can be homogenised without much blow-up in size.
The requirement of the size of the bucketing system being small also arises because of this analysis.

The only additional point that needs to be checked for the proof to go through is that similar to the commutative setting, non-commutative formulas can be depth reduced as well (\autoref{lem:depth-reduction}).
A complete proof of \autoref{thm:orderForm} can be found in \autoref{sec:orderForms}.

\subsection{Other Results: A Complete View of the Abecedarian World}

We now go over some other results that helps in completing the view of the $\abcd$ world.

As mentioned earlier, {\Hrubes} {\etal} \cite{HWY11} had defined ordered circuits, a model naturally suited to compute ordered polynomials. 
They had then gone on to show that without loss of generality, any circuit computing an ordered polynomial can be assumed to be ordered\footnote{Theorem 7.1 in \cite{HWY11}.}. 
We show that even in the $\abcd$ setting, such a statement is true.

\begin{restatable}[Converting Circuits into Abecedarian Circuits]{observation}{orderCkt}\label{obs:orderCkt}
	Let $f$ be an $\abcd$ polynomial with respect to a bucketing system of size $m$, and $\ckt$ be a circuit of size $s$ computing $f$. 
	Then there is an $\abcd$ circuit $\ckt'$ computing $f$ of size $O(m^3s)$.
\end{restatable}

\noindent What this implies is that an $n^{\omega(1)}$ lower bound against $\abcd$ circuits for any explicit polynomial that is $\abcd$ would result in a super-polynomial lower bound against general non-commutative circuits.
We also show that an analogous statement is true even for $\abcd$ ABPs.

\begin{restatable}[Converting ABPs into Abecedarian ABPs]{observation}{orderABP}\label{obs:orderABP}
	Suppose $f$ is an $\abcd$ polynomial with respect to a bucketing system of size $m$. If there is an ABP $\abp$ of size $s$ computing it, then there is an $\abcd$ ABP $\abp'$ computing it of size $O(ms)$.
\end{restatable}

\noindent Next, let us define various classes of $\abcd$ polynomials. 

Let $\abcncVP$ denote the class of $\abcd$ polynomials that can be computed by $\poly$-sized $\abcd$ circuits. 
Similarly let $\abcncVBP$ and $\abcncVF$ denote the classes of $\abcd$ polynomials that can be computed by $\poly$-sized $\abcd$ ABPs and $\abcd$ formulas respectively. 
We first note that the logical inclusions that should hold, do hold.

\begin{restatable}[The Usual Inclusions]{observation}{usualInclusions}\label{obs:usualInclusions}
	Let $\abcncVP$, $\abcncVBP$ and $\abcncVF$ denote the classes of $\abcd$ polynomials over $n$ variables that can be computed by $\poly(n)$ sized $\abcd$ circuits, $\abcd$ ABPs and $\abcd$ formulas respectively.
	Then,
	\[
		\abcncVF \subseteq \abcncVBP \subseteq \abcncVP. \qedhere	
	\]
\end{restatable}

\noindent We also observe that if a degree $d$ polynomial has an $\abcd$ ABP of size $s$, then it has an $\abcd$ formula of size $O(s^{\log d})$ via the usual divide-and-conquer algorithm. 

\begin{restatable}[Converting Abecedarian ABPs into Abecedarian Formulas]{observation}{ABPtoForm}\label{obs:ABPtoForm}
	Suppose $f$ is an $\abcd$ polynomial of degree $d$. If there is an $\abcd$ ABP $\abp$ of size $s$ computing it, then there is an $\abcd$ formula $\form$ computing $f$ of size $O(s^{\log d})$.
\end{restatable}

\noindent What \autoref{thm:main} essentially shows is that the blow-up observed in \autoref{obs:ABPtoForm} is tight.
Finally, it is not hard to see that Nisan's proof can be modified to give an exponential separation between $\abcd$ ABPs and $\abcd$ circuits. 




\subsubsection*{General Formula Lower Bound from Homogeneous Formula Lower Bound}

We end by showing that homogeneous formula lower bounds for some well-studied polynomials would lead to separating $\ncVF$ and $\ncVBP$.
These statements are corollaries of \autoref{lem:homogenisation}.

\begin{restatable}{corollary}{FormLBfromHomDetLB}\label{cor:FormLBfromHomDetLB}
	A $2^{\omega(n)}$ lower bound against homogeneous formulas computing the $n \times n$ determinant polynomial, $\Det{n}(\vecx)$, implies $\ncVF \neq \ncVBP$.
\end{restatable}

\begin{remark}
	Nisan \cite{N91} had shown that the ABP complexity of $\Det{n}(\vecx)$ and $\Perm{n}(\vecx)$ is $2^{\Theta(n)}$. 
	In the case of $\Perm{n}(\vecx)$, we also have a $2^{O(n)}$ upper bound due to Ryser \cite{R63}.
	However, to the best of our knowledge, there is no formula computing $\Det{n}(\vecx)$ of size $n^{o(n)}$. 
	Hence there is a possibility that a $2^{\omega(n)}$ lower bound can be shown against formulas for this polynomial.
\end{remark}

\begin{restatable}{corollary}{FormLBfromHomImmLB}\label{cor:FormLBfromHomImmLB}
	An $n^{\omega(1)}$ lower bound against homogeneous formulas computing the $n$-variate iterated matrix multiplication polynomial of degree $\log n$, $\imm{n}{\log n}(\vecx)$, implies a super-polynomial separation between ABPs and formulas in the non-commutative setting.
\end{restatable}

\noindent To put the requirement of degree being $O(\log n)$ in perspective, note the following.

\begin{remark}[Analogous to Remark 5.12 in \cite{LLS19}]
	The standard divide and conquer approach for computing the iterated matrix multiplication polynomial $\imm{n}{d}$ yields a (homogeneous) formula of size $n^{O(\log d)}$. 
	It would be quite surprising if this standard algorithm were not optimal in terms of formula size.

	Intuitively, improving on the standard divide and conquer algorithm gets harder as $d$ gets smaller. 
	This is because any (homogeneous) formula of size $n^{o(\log d)}$ for computing $\imm{n}{d}$ can be used in a straightforward manner to recursively obtain (homogeneous) formulas for $\imm_{n,D}$ of size $n^{o(logD)}$ for any $D > d$. 
	The case of smaller $d$, which seems harder algorithmically, is thus a natural first candidate for lower bounds.
\end{remark}

\subsection{Structure of the Paper}

We begin, in \autoref{sec:definitions}, with formal definitions for $\abcd$ polynomials and naturally restricted version of circuits, ABPs and formulas that compute them.
Then, in \autoref{sec:structural-statements}, we prove some structural statements, namely \autoref{lem:depth-reduction} and \autoref{lem:homogenisation}.
In \autoref{sec:orderModels}, we prove \autoref{thm:orderForm} along with \autoref{obs:orderCkt} and \autoref{obs:orderABP}.
We then prove our main result (\autoref{thm:main}), that gives a super-polynomial separation between $\abcd$ formulas and ABPs, in \autoref{sec:main}.
Finally, in \autoref{sec:simpleStatements}, we prove the remaining statements mentioned above.


\section{Preliminaries}\label{sec:definitions}

Let us begin by formally defining $\abcd$ polynomials and the naturally restricted versions of circuits, ABPs and formulas that compute them.

\subsection{Abecedarian Polynomials}\label{sec:defn_abcd_poly}

First, we define the notion of a bucketing system.

\begin{definition}[Bucketing System for a Set of Indeterminates]
	Suppose $\set{x_1, \ldots, x_n}$ is a set of indeterminates.
	An ordered set $(X_1, \ldots, X_m)$ is said to be a bucketing system for it, if $\set{X_1, \ldots, X_m}$ forms a partition of the set $\set{x_1, \ldots, x_n}$.

	In this case, we may also denote the bucketing system $(X_1, \ldots, X_m)$ by $\set{X_1, \ldots, X_m}$, since the ordering among $X_1, \ldots X_m$ is clear from context.	
\end{definition}

\noindent Next, let us formally define $\abcd$ polynomials. 

\begin{definition}[Abecedarian Polynomials]\label{def:ordPoly}
	A polynomial $f \in \F\inangle{x_1, \ldots, x_n}$ of degree $d$ is said to be $\abcd$ with respect to a bucketing system $\set{X_1, \ldots, X_m}$ for $\set{x_1, \ldots x_n}$, if  
	\[
		f = f[\emptyset) + \sum_{k=1}^{d} \inparen{\sum_{1 \leq i_1 \leq \cdots \leq i_k \leq m} f[X_{i_1}, \ldots, X_{i_k}]}
	\] 
	where $f[\emptyset)$ is the constant term in $f$, and for any $k \in [d]$, $f[X_{i_1}, \ldots, X_{i_k}]$ is defined as follows. 
	For a polynomial $f$, $f[X_{i_1}, \ldots, X_{i_k}]$ is the homogeneous polynomial of degree $k$ such that for every monomial $\alpha$,
	\[
		\hspace{-1em} \coeff_\alpha(f[X_{i_1}, \ldots, X_{i_k}]) =
		\begin{cases}
			\coeff_\alpha(f) & \text{ if } \alpha = x_{\ell_1} \cdots x_{\ell_k} \text{ with } x_{\ell_j} \in X_{i_j} \text{ for every } j \in [k]\\
			0 & \text{ otherwise.} 
		\end{cases}
	\]
	In this case, we say that $f$ is $\abcd$ with respect to $\set{X_1, \ldots, X_m}$, a bucketing system of size $m$.
\end{definition}

Abecedarian polynomials are essentially generalisations of ordered polynomials (defined by {\Hrubes}, Wigderson and Yehudayoff \cite{HWY11}).
A homogeneous polynomial, of degree $d$, is said to be ordered if the set of variables it depends on can be partitioned into $d$ buckets such that variables occuring in position $k$ only come from the $k$-th bucket.

It is easy to see that any ordered polynomial is also $\abcd$ with respect to the same bucketing system.
This is because position indices are always increasing. 
For example, consider the following version of the \emph{complete homogeneous symmetric polynomial}.
\[
	\ochsym{n}{d} (\vecx) = \sum_{1 \leq i_1 \leq \ldots \leq i_d \leq n} x^{(1)}_{i_1} \cdots x^{(d)}_{i_d}.
\]
It is both ordered as well as $\abcd$ with respect to the buckets $\set{X_k = \setdef{x^{(k)}_i}{i \in [n]}}$.

However, note that there are homogeneous polynomials that are $\abcd$ but not ordered.
The following version of the same polynomial is an example.
\[
	\chsym{n}{d} (\vecx) = \sum_{1 \leq i_1 \leq \ldots \leq i_d \leq n} x_{i_1} \cdots x_{i_d}
\]
is $\abcd$ with respect to $\setdef{X_i}{X_i = \set{x_i}}$, but is not ordered.

The reason is that for a polynomial to be ordered, the bucket labels have to essentially be position labels.
On the other hand, for a polynomial to be $\abcd$ with respect to a bucketing system, the bucket labels can be independent of position.

For example, note that $\ochsym{n}{d}(\vecx)$ is $\abcd$ with respect to the bucketing system $\set{X_i = \setdef{x^{(k)}_i}{k \in [d]}}$ along with the one mentioned earlier.\\

\noindent We now move on to defining algebraic models that naturally compute $\abcd$ polynomials.

\subsection{Abecedarian Circuits}\label{sec:defn_abcd_ckts}

Homogeneous formulas have the property that any vertex can be labelled by a tuple of position indices $(a, b)$ such that all the monomials being computed at that vertex occur exactly from position $a$ to position $b$ in the final polynomial that is being computed by it.
\Hrubes $\text{}$ \etal $\text{}$ \cite{HWY11} defined \emph{ordered} circuits to be those circuits that have this property.\\

\noindent A circuit computing a degree $d$ polynomial $f \in \F \inangle{x_1, \ldots, x_n}$ is said to be ordered, if $\set{X_1, \ldots, X_d}$ forms a partition of $\set{x_1, \ldots, x_n}$ such that  
\begin{itemize}
	\item every gate $v$ in the circuit is labelled by a tuple of position indices $(a,b)$;
	\item if $f_v$ is the polynomial computed at $v$, then
	\begin{itemize}
		\item $f_v$ is homogeneous and has degree $(b-a+1)$;
		\item every monomial in $f_v$ is a product of exactly one variable from each of the buckets $X_a, \ldots, X_b$, multiplied in increasing order of their bucket indices.
	\end{itemize}
\end{itemize}

\noindent We generalise this notion to define circuits that naturally compute $\abcd$ polynomials.
Before we can do that, we need the notion of sub-polynomials of any $\abcd$ polynomial. 

\begin{definition}[Sub-Polynomials of an Abecedarian Polynomial]\label{def:subPoly-of-OrdPoly}
	Suppose $f$ is an $\abcd$ polynomial with respect to the bucketing system $\set{X_1, \ldots, X_m}$, and has degree $d$. For any $1 \leq a \leq b \leq m+1$, $f[a,b)$ is the sub-polynomial of $f$ defined as follows.
	\begin{itemize}
		\item For any $a \in [m+1]$, $f[a, a) = f[\emptyset)$ is the constant term in $f$.
		\item For any $1 \leq a < b \leq m+1$,
		\[
			f[a,b) = \sum_{k=1}^{d} \inparen{\sum_{\substack{i_1, \ldots, i_k \in [m] \\ a = i_1 \leq \cdots \leq i_k < b}} f[X_{i_1}, \ldots, X_{i_k}]}	
		\]
		where $f[X_{i_1}, \ldots, X_{i_k}]$ is as defined in \autoref{def:ordPoly}.
	\end{itemize}
	Further, we say that a polynomial $f$ is of \emph{type} $[a,b)$ if $f = f[a,b)$.
\end{definition}

\noindent Let us now formally define $\abcd$ circuits.

\begin{definition}[Abecedarian Circuits]\label{def:ordCkt}
	For any $a, b \in \N$, let $[a,b)$ denote a set of the form $I = \setdef{i}{a \leq i < b}$. 
	As a convention, $[a,a)$ denotes the empty set for every $a \in \N$. 
	
	A multi-output circuit $\ckt$ is said to be $\abcd$ when
	\begin{itemize}
		\item every gate $v$ in $\ckt$ is associated with a set $I_v = [a,b)$;
		\item if $f_v$ is the polynomial computed at $v$, then $f_v = f[a, b)$;
		\item if $v = v_1 + v_2$, then $I_v = I_{v_1} = I_{v_2}$;
		\item if $v = v_1 \times v_2$ with $I_v = [a, a)$, then $I_{v_1} = I_{v_2} = [a,a)$
		\item  if $v = v_1 \times v_2$ with $I_v = [a, b)$ and $a < b$, then one of the following is true
		\begin{itemize}
			\item $I_{v_1} = [a,b)$ and $I_{v_2} = [b,b)$;
			\item $I_{v_1} = [a,a)$ and $I_{v_2} = [a,b)$;
			\item there exists $a \leq c < b$ such that $I_{v_1} = [a, c+1)$ and $I_{v_2} = [c, b)$.
		\end{itemize} 
	\end{itemize}
	The polynomial computed by $\ckt$ is the sum of the polynomials computed at the different output gates.
\end{definition}

\hspace{.5em}

\noindent Next, we define $\abcd$ ABPs and $\abcd$ formulas as the restricted versions of ABPs and formulas respectively, that naturally compute $\abcd$ polynomials.

\subsection{Abecedarian ABPs and Formulas}\label{sec:defn_abcd_ABP_form}

Homogeneous ABPs have the property that every vertex in it is labelled by a position index such that, polynomials computed between vertices labelled with indices $a$ and $b$ only contain monomials between positions $a$ and $(b-1)$.

We define $\abcd$ ABPs analogously except that the labels on the vertices are bucket labels instead of position labels. 
These restricted ABPs naturally compute $\abcd$ polynomials.

\begin{definition}[Abecedarian ABPs]\label{def:ordABP}
	A multi-input, multi-output ABP $\abp$ is said to be $\abcd$ when
	\begin{itemize}
		\item every vertex in it is labelled by a bucket index;
		\item if $f$ is the polynomial computed between vertices labelled with indices $a$ and $b$ respectively, then $f = f[a,b+1)$. 
	\end{itemize}
	The polynomial computed by $\abp$ is the sum of all the polynomials computed between the various (input, output) gate pairs.
\end{definition}

Similarly, we define $\abcd$ formulas as analogues of homogeneous formulas, with the labels again referring to bucket indices instead of position indices.

\begin{definition}[Abecedarian Formulas]\label{def:ordForm}	
	Let sets of the form $[a,b)$, with $a,b \in \N$, be as defined in \autoref{def:ordCkt}.
	Suppose $\form$ is a formula computing a polynomial $f$ that is $\abcd$ with respect to a bucketing system of size $m$. Then $\form$ is said to be $\abcd$ if $\form$ is in fact a collection of formulas $\setdef{\form_i}{i \in [m]}$ such that for every $i \in [m]$,
	\begin{itemize}
		\item $\form_i$ computes the polynomial $f[i,m+1)$;
		\item every gate $v$ in $\form_i$ is associated with a set $I_v = [a,b)$, and in particular, the root node must be associated with the set $[i, m+1)$
		\item if $f_v$ is the polynomial computed at $v$, then $f_v = f_v[a, b)$;
		\item if $v = v_1 + v_2$, then $I_v = I_{v_1} = I_{v_2}$;
		\item if $v = v_1 \times v_2$ with $I_v = [a, a)$, then $I_{v_1} = I_{v_2} = [a,a)$
		\item  if $v = v_1 \times v_2$ with $I_v = [a, b)$ and $a < b$, then one of the following is true
		\begin{itemize}
			\item $I_{v_1} = [a,b)$ and $I_{v_2} = [b,b)$;
			\item $I_{v_1} = [a,a)$ and $I_{v_2} = [a,b)$;
			\item there exists $a \leq c < b$ such that $I_{v_1} = [a, c+1)$ and $I_{v_2} = [c, b)$.
		\end{itemize} 
	\end{itemize}
	The polynomial computed by $\form$ is the sum of the polynomials computed by the various $\form_i$s.
	Further, $\form$ is said to be homogeneous if each $\form_i$ is homogeneous.
\end{definition}

\hspace{.5em}

\noindent With these definitions in mind, we now move to proving some structural statements. 

\section{Structural Statements}\label{sec:structural-statements}

In this section, we prove two structural statements in the non-commutative setting that are known to be true in the commutative setting. 
Apart from being crucial to our proofs, they are possibly interesting observations in their own right.

\subsection{Depth Reduction for Non-Commutative Formulas}\label{sec:depth-reduction}

Brent \cite{B74} had shown that if there is a formula of size $s$ computing a commutative polynomial $f$, then there is a formula of depth $O(\log s)$ and size $\poly(s)$ that computes the same polynomial. 
We show that this is also true in the non-commutative setting.

The proof is essentially the same as the one by Brent \cite{B74}, just analysed carefully.
We give the complete proof for the sake of completeness.

\depthReduction*

\begin{proof} 
	Suppose $\form$ is a fan-in $2$ formula of size $s$ that computes $f$. 
	Then, we claim the following.

\begin{claim}\label{clm:depth-reduction_induction-step}
	Suppose $\form_0$ is a formula computing a polynomial $f_0$ and has fan-in 2. 
	Then the there exist sub-formulas, $L, \form_1, R, \form_2$, of $\form_0$ such that
	\begin{itemize}
		\item $\form'_0 = L \cdot \form_1 \cdot R + \form_2$ also computes $f_0$;
		\item each of $L, \form_1, R, \form_2$ have size at least $(s/3)$ and at most $(2s/3)$;
		\item if $\form_0$ is homogeneous, then so are $L, \form_1, R, \form_2$;
		\item if $\form_0$ is $\abcd$ with respect to some bucketing system,  $f_\l$, $f_1$, $f_\r$, $f_2$ are polynomials computed by $L, \form_1, R, \form_2$ respectively and $f_0 = f_0[a,b)$, then $f_2 = f_2[a,b)$ and
		\begin{itemize}
			\item each of $L, \form_1, R, \form_2$ are $\abcd$ with respect to the same bucketing system as $\form_0$
			\item when $a = b$, $\qquad f_\l = f_\l[a,a) \qquad f_1 = f_1[a,a) \qquad f_\r = f_\r[a,a)$;
			\item when $a<b$, there exist $a \leq i \leq j \leq b$ such that
			
			\vspace{-2.5em}
			\begin{multline*}
				\hspace{-3em}
				a = i < j = b \quad \implies \quad f_\l = f_\l[a,i) \qquad f_1 = f_1[i,j) \qquad f_\r = f_\r[j,b).\\
				\hspace{-3em}
				a = i = j < b \quad \implies \quad f_\l = f_\l[a,i) \qquad f_1 = f_1[i,j) \qquad f_\r = f_\r[j,b).\\
				\hspace{-3em}
				a = i < j < b \quad \implies \quad f_\l = f_\l[a,i) \qquad f_1 = f_1[i,j+1) \qquad f_\r = f_\r[j,b).\\
				\hspace{-3em}
				a < i = j = b \quad \implies \quad f_\l = f_\l[a,i+1) \qquad f_1 = f_1[i,j) \qquad f_\r = f_\r[j,b).\\
				\hspace{-3em}
				a < i = j < b \quad \implies \quad f_\l = f_\l[a,i+1) \qquad f_1 = f_1[i+1,j+1) \qquad f_\r = f_\r[j,b).\\
				\hspace{-3em}
				a < i < j = b \quad \implies \quad f_\l = f_\l[a,i+1) \qquad f_1 = f_1[i,j) \qquad f_\r = f_\r[j,b).\\
				\hspace{-7cm}
				a < i < j < b \quad \implies \quad f_\l = f_\l[a,i+1) \qquad f_1 = f_1[i,j+1) \qquad f_\r = f_\r[j,b).
			\end{multline*}
		\end{itemize}
	\end{itemize} 
\end{claim}

\noindent Before proving \autoref{clm:depth-reduction_induction-step}, let us complete the proof of \autoref{lem:depth-reduction} using it.

By the above claim, we have a formula $\form'_0$ computing  $f_0$ that looks like $L \cdot \form_1 \cdot R + \form_2$ where each of $L, \form_1, R, \form_2$ have size at most $(2s/3)$. 
Further if $\form$ is homogeneous, then so are each of $L, \form_1, R, \form_2$.
Hence, $\form'_0$ is homogeneous.
On the other hand, when $\form_0$ is $\abcd$, so are $L, \form_1, R, \form_2$. 
Further, note that $\form'_0$ is also $\abcd$ in this case since $f_\l, f_1, f_\r, f_2$ are of the \emph{correct type} due to \autoref{clm:depth-reduction_induction-step}.\\ 

\noindent In all the cases, recursively applying this technique, on each of $L, \form_1, R, \form_2$, we get
\[
	\depth(s) \leq \depth(2s/3) + 3 \text{\qquad and \qquad} \size(s) \leq 4 \cdot \size(2s/3) + 3.
\]
Note that in the base case, when $s$ is constant, both $\size(s)$ and $\depth(s)$ are constants. Thus,
\[
	\depth(s) = O(\log s) \text{\qquad and \qquad} \size(s) = \poly(s). \qedhere
\]
\end{proof}

\noindent Pictorially, once we have \autoref{clm:depth-reduction_induction-step}, we essentially do the following recursively.

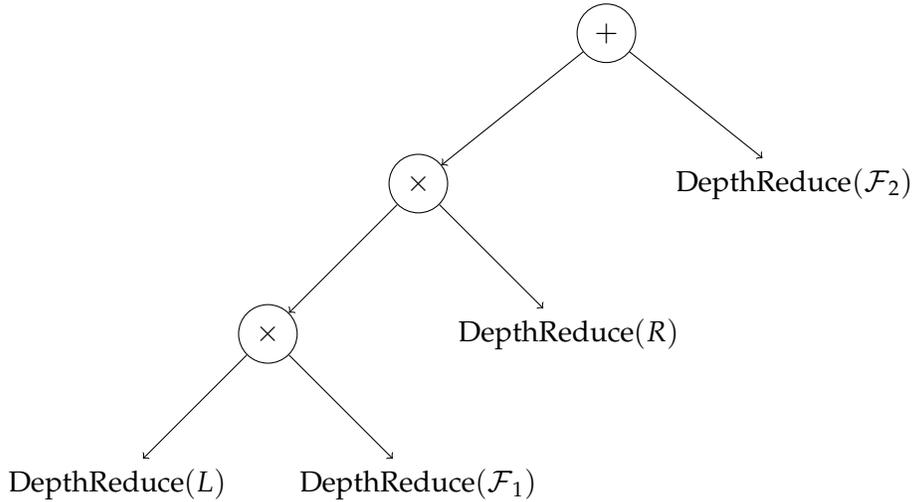
\vspace{1em}
\begin{tikzpicture}
	\node (L) at (-4,-4) {$\operatorname{DepthReduce}(L)$};
	\node (v) at (0,-4) {$\operatorname{DepthReduce}(\form_1)$};
	\node[circle, draw=black] (times1) at (-2,-2) {$\times$}
	edge[->] (L)
	edge[->] (v);

	\node (R) at (2,-2) {$\operatorname{DepthReduce}(R)$};

	\node[circle, draw=black] (times2) at (0,0) {$\times$}
	edge[->] (times1)
	edge[->] (R);
	\node (Fv0) at (5,0) {$\operatorname{DepthReduce(\form_2)}$};
	\node[circle, draw=black] (root) at (2.5,2) {$+$}
	edge[->] (times2)
	edge[->] (Fv0);
  \end{tikzpicture}

\vspace{1em}
\noindent We now complete the proof of \autoref{clm:depth-reduction_induction-step}

\begin{proofof}{\autoref{clm:depth-reduction_induction-step}}
	From the root let us traverse $\form_0$ towards the leaves, always choosing the child that has a larger sub-tree under it, till we find a vertex $v$ such that the associated sub-tree has size at most $(2s/3)$. 
	Since $\form_0$ tree has fan-in 2, we also know that the size of this sub-tree must be at least $(s/3)$. 
	Let this sub-tree be $\form_1$. 
	Additionally, in the case when $\form_0$ is $\abcd$, let us assume that $v$ is labelled with $[i_v,j_v)$.
	
	Let $\mathcal{P}$ be the path from $v$ to the root and $v_{\text{add}}$ the addition gate on $\mathcal{P}$ which is closest to $v$.
	Also let the set of multiplication gates on $\mathcal{P}$ be $\set{v_1, \ldots, v_\ell}$ for some $\ell \in \N$.
	Assume, without loss of generality, that $v_1$ is closest to $v$ and $v_\ell$ to the root.
	Further, for every $i \in [\ell]$, let $L_i$ be sub-formula corresponding to the left child of $v_i$ and $R_i$ the one to its right child.
	Note that for every $i \in [\ell]$, exactly one of children of $v_i$ is a vertex in $\mathcal{P}$.
	We can then define $L$ and $R$ as follows.
	\begin{description}
		\item[Step 1:] Set $L = R = 1$.
		\item[Step 2:] For $i$ from $1$ to $\ell$,
		    \[
		        L =
		        \begin{cases}
                   	L_i \times L & \text{if the right child of $v_i$ is a vertex in $\mathcal{P}$,}\\
                  	L & \text{otherwise.}
                \end{cases}
            \]
            and
            \[
            	R = 
            	\begin{cases}
                   	R & \text{if the right child of $v_i$ is a vertex in $\mathcal{P}$,}\\
                  	R \times R_i & \text{otherwise.}
                \end{cases}
            \]
	\end{description}
	
	\noindent Also define $\form_2$ to be the formula we get by replacing the vertex $v$ and the sub-tree under it with $0$, and then removing the redundant gates.
		
	Clearly, by construction, $\form_1$, $L$, $R$ and $\form_2$ are sub-formulas of $\form_0$.
	Further, $\form_1$ is disjoint from $L$, $R$ and $\form_2$. 
	As a result, since $\form_1$ has size at least $(s/3)$ and at most $(2s/3)$, it must be the case that each of $L$, $R$ and $\form_2$ have size at least $(s/3)$ and at most $(2s/3)$.

	Also, it is not hard to see that $\form'_0 = L \cdot \form_1 \cdot R + \form_2$ computes $f_0$. 
	What is left to check is that when $\form_0$ is homogeneous or $\abcd$, then $L, \form_1, R, \form_2$ have the additional properties claimed.
	The one line proof of this is that each \emph{parse-tree}\footnote{For a definition, see for example \cite{LLS19}} of $\form_0$ is merely restructured in the above process, without changing its value.
	We however go over the proof explicitly for the sake of completeness.
	
	When $\form_0$ is homogeneous, since $L, \form_1, R, \form_2$ are sub-formulas, they are also homogeneous.

	\noindent On the other hand, suppose $\form_0$ is $\abcd$ and $f_0 = f_0[a,b)$.
	Recall that the vertex $v$ was labelled by $[i_v, j_v)$.
	Let us set $i = i_v$ and $j=j_v$. 
	Then, by definition, $\form_1$ is  labelled by $[i,j)$.
	Hence, if $f_1$ is the polynomial computed at $v$, then $f_1 = f_1[i,j)$.
	Further, $\form_1$ is $\abcd$ since it is a sub-formula of $\form_0$ and computes an $\abcd$ polynomial.

	Now let us focus on $\form_2$.
	Essentially $\form_2$ is got by removing from $\form_0$, $v$ and all the multiplication gates on $P$ between $v$ and $v_{\text{add}}$ along with the sub-trees under them.
	Thus $\form_2$ is also $\abcd$ in this case, and if $f_2$ is the polynomial by it, then $f_2 = f_2[a,b)$.
	
	Finally, note that the left indices of labels on the various vertices of $\mathcal{P}$ change only at the gates at which multiplications to $L$ occur.
	Further, note that they occur in the \emph{correct order} and are of the \emph{correct type}.
	Thus, by induction, it is easy to see that the labels on $L$ are consistent with those on the $L_i$s when the respective multiplications happen.
	Therefore $L$ is $\abcd$, and $f_\l = f_\l[a,i)$.	
	
	For similar reasons, $R$ is also $\abcd$ and $f_\r = f_\r[j,b)$.
	This completes the proof.
\end{proofof}

\subsection{Homogenisation}\label{sec:homogenisation}

Raz \cite{R13} had shown that if there is a formula computing a homogeneous polynomial of \emph{low} degree in the commutative world, then it can be assumed without loss of generality that the formula is homogeneous.
We show that his proof also works in the non-commutative setting because of \autoref{lem:depth-reduction}. 
A complete proof is given here for the sake of completeness.

\homogenisation*

\begin{proof}
	We first note that since $s$ is the ABP complexity of $f$, $s' \geq s$.
	Further if $\form$ has depth $r$, then by \autoref{lem:depth-reduction}, we can assume without loss of generality, that $r = O(\log s')$.
	
	In order to construct a homogeneous formula computing $f$, we first homogenise $\form$ to obtain a circuit $\ckt$, and then \emph{unravel} $\ckt$ to make it into a formula $\form'$.	
	
	The first step is done in the usual manner.
	For every gate $v$ in $\form$, we have $d+1$ gates $(v,0)$, $\ldots$, $(v,d)$ in $\ckt$.
	Intuitively if $f_v$ is the polynomial computed at $v$, then the polynomial computed at $(v,i)$ is the degree $i$ homogeneous component of $f_v$. 
	These vertices are then connected as follows.
	
	\begin{itemize}
		\item If $v = u_1 + u_2$, then for every $i \in \set{0, \ldots, d}$, \qquad $(v, i) = (u_1, i) + (u_2, i)$.
		\item If $v = u_1 \times u_2$, then for every $i \in \set{0, \ldots, d}$, \qquad $(v, i) = \sum_{j=0}^{i} (u_1, j) \times (u_2, i-j)$.
	\end{itemize}

	So, we now have a homogeneous circuit $\ckt$ that computes $f$ and has size at most $O(d^2 \cdot s')$.
	Also, the depth of this circuit is at most twice that of $\form$, and the multiplication gates have fan-in $2$.
	
	To convert $\ckt$ into a formula $\form'$, we have to recompute nodes whenever they have to be reused. 
	That is, a particular vertex in $\ckt$ has to be duplicated as many times as there are paths from the vertex to the root. 
	Thus, to upper bound the size of $\form'$, we need to give an upper bound on the number of distinct paths from every vertex of $\ckt$ to its root.
	
	Let us arbitrarily choose a vertex $(v,i)$ in $\ckt$, and consider the path from it to the root. 
	Suppose the path is $(v,i) = (v_1, i_1) \rightarrow \cdots \rightarrow (v_\ell, i_\ell) = (\rt, d)$ where $\ell$ is at most the depth of $\ckt$. 
	Note that it must be the case that $i = i_1 \leq \cdots \leq i_\ell = d$. 
	Hence, if we define $\delta_j = i_{j+1} - i_j$ for $j \in [\ell-1]$, then the $\delta_j$s are non-negative integers such that $\delta_1 + \cdots + \delta_{\ell-1} = (d - i)$. 
	Thus, the number of  choices we have for $(i_2, \ldots, i_\ell)$ such that $i = i_1 \leq \cdots \leq i_\ell = d$, is the same as the number of choices we have for $(\delta_1, \ldots, \delta_{\ell-1})$ such that $\delta_1 + \cdots + \delta_{\ell -  1} = (d - i) \leq d$. 
	This is at most ${\ell + d}\choose{\ell}$.
	
	Note that in this process the fan-in of the gates have not changed, and hence the multiplication gates in $\form'$ continue to have fan-in $2$.
	Further, we know that the $\ckt$ has depth $2r$ and hence $\ell \leq 2r$. 
	Therefore, the number of paths from $(v, i)$ to the root is at most $2r + d \choose 2r$. 
	Hence, if $\form'$ is the formula obtained by unravelling $\ckt$, then $\size(\form') \leq s' \cdot d^2 \cdot {2r + d \choose d}$.	
	Here $r = O(\log(s'))$, and $s \leq s'$ implying that $d = O(\log(s)) = O(\log(s'))$.
	Thus, $\size(\form') \leq \poly(s')$.

	Finally, assume that $\form$ is $\abcd$.
	Then every vertex $v$ is labelled with a tuple of bucket indices, say $(a_v,b_v)$.
	In that case, we add the label $(a_v,b_v)$ to the gates $\set{(v,i)}_{i=0}^{d}$ in $\ckt$ and continue with the proof as is.
	Note that the final formula that we get, $\form'$, is $\abcd$ and all the other properties that were true in the general case, continue to be true.
\end{proof}

\section{Converting Computational Models into Abecedarian Ones}\label{sec:orderModels}

In this section we show that, without loss of generality, circuits and ABPs computing $\abcd$ polynomials can be assumed to be $\abcd$.
For formulas however, we can prove such a statement only in certain cases.

\subsection{Circuits}\label{sec:orderCkts}
{\Hrubes} {\etal} \cite{HWY11} had shown that any circuit computing an ordered polynomial can be assumed to be ordered without loss of generality. 

\begin{theoremwp}[Theorem 7.1 in \cite{HWY11}]\label{thm:orderHomCkts}
	Let $\ckt$ be a circuit of size $s$ computing an ordered polynomial $f$ of degree $d$. 
	Then, there is an ordered circuit $\ckt'$ of size $O(d^3s)$ that computes $f$.
\end{theoremwp}

\noindent We show that the proof of this statement can be generalised to show \autoref{obs:orderCkt}. 
A complete proof is given for the sake of completeness.

\orderCkt*

\begin{proof}
	Without loss of generality, let us assume that $\ckt$ has fan-in $2$.
	
	We prove the given statement by describing how to construct $\ckt'$ from $\ckt$. 	
	For each gate $v$ in $\ckt$, we make $O(m^2)$ copies in $\ckt'$, $\setdef{(v, [a,b))}{1 \leq a \leq b \leq m+1}$; and if $\rt$ is the output gate in $\ckt$, then we define the set of output gates in $\ckt'$ to be $\set{(\rt, [i,m+1))}_{i \in [m+1]}$. 
	
	Intuitively, if $f_v$ is the polynomial computed at $v$ in $\ckt$, then the polynomial computed at $(v, [a,b))$ is $f_v[a,b)$. 
	Thus if $f$ was the polynomial computed at $\rt$, then the polynomial computed by $\ckt'$ is $\sum_{i=1}^{m+1} f[i,m+1)$ which is indeed $f$.
	
	We ensure this property at every gate by adding edges as follows. 	
	\begin{itemize}
		\item If $v$ is an input gate labelled by a field element $\gamma$, 
		\begin{itemize}
			\item we set $(v, [a,a)) = \gamma$ for every $a \in [m+1]$;
			\item we set $(v, [a,b)) = 0$ for every $1 \leq a < b \leq m+1$.
		\end{itemize}
		\item If $v$ is an input gate labelled by a variable $x_i$ and $x_i \in X_k$, 
		\begin{itemize}
			\item we set $(v, [k, k+1)) = x_i$;
			\item we set $(v, [a,b)) = 0$ for every $a \neq k$, $b \neq k+1$.
		\end{itemize} 
		\item If $v = v_1 + v_2$, \qquad we set $(v, [a,b)) = (v_1, [a,b)) + (v_2, [a,b))$ for every $a \leq b \in [m+1]$.
		\item If $v = v_1 \times v_2$, \qquad we set $(v, [a,a)) = (v_1, [a,a)) \cdot (v_2, [a,a))$ for every $a \in [m+1]$; \quad and 
		\[
			\hspace{-2.5em}
			(v, [a,b)) = (v_1, [a,a)) \cdot (v_2, [a,b)) + (v_1, [a,b)) \cdot (v_2, [b,b)) + \sum_{c=a}^{b-1} (v_1, [a, c+1)) \times (v_2, [c, b))
		\]
		for every $1 \leq a < b \leq m+1$.
	\end{itemize}
	\noindent Finally, for every $1 \leq a \leq b \leq m+1$, we associate the gate $(v, [a,b))$ in $\ckt'$ with the set $[a,b)$. 
	
	Using induction, one can easily show that the gates in $\ckt'$ have the claimed properties. Hence $\ckt'$ is indeed an $\abcd$ circuit computing $f$. 
	Further for every gate $v$ in $\ckt$, there are at most $O(m^3)$ vertices in $\ckt'$. Thus the size of $\ckt'$ is $O(m^3 s)$.
\end{proof}

\subsection{Algebraic Branching Programs}\label{sec:orderABPs}

Next, we show that a similar statement is true for ABPs as well.

\orderABP*

\begin{proof}
	Let $f$ have degree $d$ and be $\abcd$ with respect to the buckets $\set{X_i}_{i=1}^{m}$, where $X_i = \setdef{x_{i,j}}{j \in [n_i]}$.	
	Without loss of generality, we can assume that $\abp$ is homogeneous\footnote{Every edge is labelled by a homogeneous form.}.
	If $f$ is not homogeneous, $\abp$ can be thought of as a collection of homogeneous ABPs $\set{\abp_1, \ldots, \abp_d}$ where $\abp_k$ computes the $k$-th homogeneous component of $f$. 
	
	We prove the theorem by describing how to construct $\abp'$. 
	For each vertex $v$ in $\abp$, make $O(m)$ copies in $\abp'$, namely $\setdef{(v, a)}{0 \leq a \leq m}$. 
	Intuitively, if $g_{(u,v)}$ is the polynomial computed between $u$ and $v$ in $\abp$, then the polynomial computed between $(u, a)$ and $(v, b)$ in $\abp'$ is $g_{(u,v)}[a,b+1)$.
	The way we ensure this property at every vertex is by adding edges in $\abp'$ as follows.
	
	\begin{quote}
		For any two vertices $u$, $v$ in $\abp$, suppose there is an edge between them that is labelled with $\sum_{i \in [m]} \sum_{j \in [n_i]} \gamma_{i,j} x_{i,j}.$ 
		Then, for every $a, b \in [m]$ with $a \leq b$, add an edge from $(u, a)$ to $(v, b)$ with label $\sum_{i = a}^{b} \inparen{\sum_{j \in [n_i]} \gamma_{i,j} x_{i,j}}$.
	\end{quote}

	\noindent Also, associate the bucket index $a$ with the gate $(v, a)$ in $\abp'$. 
	
	By induction, one can easily show that the gates in $\abp'$ have the claimed property.
	Hence $\abp'$ is indeed an $\abcd$ ABP computing $f$.
	Further, every vertex $v$ in $\abp$, there are at most $O(m)$ vertices in $\abp'$.
	Therefore, the size of $\abp'$ is $O(ms)$.
\end{proof}

\subsection{Formulas}\label{sec:orderForms}

Finally we show that in the case of formulas, we can prove a similar statement only when the polynomial is $\abcd$ with respect to a bucketing system of \emph{small size}.
The proof is very similar to that of \autoref{lem:homogenisation}.

\orderForm*

\begin{proof}
	Let us assume additionally that $\form$ has depth $r$. 
	Now \autoref{lem:depth-reduction} implies that $r = \log(s)$ without loss of generality.
	By \autoref{obs:orderCkt}, there is an $\abcd$ circuit $\ckt$ that computes $f$ and has size at most $s' = O(s \cdot m^3)$. 
	Further its proof implies that the depth of $\ckt$ is at most $2r$.
	
	To convert $\ckt$ into an $\abcd$ formula $\form'$, we have to recompute a node each time it has to be reused. 
	That is, a particular vertex in $\ckt$ has to be duplicated as many times as there are paths from the vertex to the root. 
	Thus to upper bound the size of $\form'$, we need to give an upper bound on the number of distinct paths from every vertex in $\ckt$ to its root.
	
	Let us arbitrarily choose a vertex $(v,[a,b))$ in $\ckt$, and consider the path from it to the root.
	Suppose the path is $(v,[a,b)) = (v_1, [a_1, b_1)) \rightarrow \cdots \rightarrow (v_\ell, [a_\ell, b_\ell)) = (\rt, [i, m+1))$ for some $\ell$ that is at most the depth of $\ckt$. 
	Note that it must be the case that 
	\[
		i \leq a_\ell \leq \cdots \leq a_1 \leq a \leq b \leq b_1 \leq b_\ell \leq m+1.
	\]

	Let us define $\delta_j = a_{j} - a_{j+1}$ and $\delta'_{j} = b_{j+1} - b_{j}$ for $j \in [\ell - 1]$. 
	Then, the number of choices we have for $(a_1, \ldots, a_\ell)$ and $(b_1, \ldots, b_\ell)$ such that 
	\[
		i = a_\ell \leq \cdots a_1 = a \leq b = b_1 \leq \cdots \leq b_\ell = m+1
	\]
	is the same as the number of choices we have for $(\delta_1, \ldots, \delta_{\ell-1}, \delta'_1, \ldots, \delta'_{\ell - 1})$ such that 
	\[
		\delta_1 + \cdots + \delta_{\ell - 1} + \delta'_1 + \cdots + \delta'_{\ell - 1} = (m+1 - (b-a) - i) \leq m.	
	\] 
	This is clearly at most ${2\ell + m}\choose{m}$.
	
	Further, we know that the $\ckt$ has depth $2r$ and hence $\ell \leq 2r$. 
	Therefore, the number of paths from $(v, i)$ to the root is at most $4r + m \choose m$. 
	Hence if $\form'$ is the formula obtained by unravelling $\ckt$, then $\size(\form') \leq s' \cdot m^2 \cdot {4r + m \choose m}$.	
	Here $s'  = O(m^3 \cdot s)$, $r = O(\log(s))$ and $m = O(\log(s))$.
	Thus, $\size(\form') \leq \poly(s)$.
\end{proof}

\section{Separating Abecedarian ABPs and Abecedarian Formulas}\label{sec:main}

In this section, we prove our main theorem: a super-polynomial separataion between the powers of $\abcd$ formulas and ABPs.
Before proceeding to the proof however, we first go over some observations that will help us with the proof.

\subsection{Some Simple Observations}

The two main polynomials we will be working with are $\lchsym{n}{d}$ and $\chsym{n}{d}$.
Let us recall their definitions.
\[
	\lchsym{n}{d}(\vecx) = \sum_{i_0=1}^{n} \inparen{\sum_{i_0 \leq i_1 \leq \ldots \leq i_d \leq n} x_{i_0, i_1} \cdot x_{i_1,i_2} \cdots x_{i_{d-1}, i_d}},
\]
is $\abcd$ with respect to the bucketing system $\set{X_1, \ldots, X_n}$ where $X_i = \setdef{x_{i,j}}{j \in [n]}$, and
\[
	\chsym{n}{d}(\vecx) = \sum_{1 \leq i_1 \leq \ldots \leq i_d \leq n} x_{i_1} \cdots x_{i_d}.	
\]
is $\abcd$ with respect to the bucketing system $\setdef{X_i}{X_i = \set{x_i}}$.\\

\noindent We begin with the notion of a \emph{linked} $\abcd$ formula computing $\lchsym{n}{d}(\vecx)$.

\begin{definition}
	An $\abcd$ formula computing $\lchsym{n}{d}$ is said to be \emph{linked} if at every gate, all the monomials occuring in the polynomial computed at that gate has the following property.
	\[
		x_{ij} \text{ appears right before } x_{i'j'} \text{ in the monomial} \implies j=i'. \qedhere	
	\]
\end{definition}

\noindent The first observation shows that any $\abcd$ formula computing $\lchsym{n}{d}(\vecx)$ can be assumed to be \emph{linked} without loss of generality.

\begin{observation}\label{obs:orderInterval}
	Let $\form$ be a homogeneous $\abcd$ formula computing $\lchsym{n}{d}(\vecx)$ of size $s$, and let the multiplication gates of $\form$ have fan-in $2$. 
	Then there is a homogeneous \emph{linked} $\abcd$ formula $\form'$ computing the same polynomial of size $O(s)$.
\end{observation}

\begin{proof}
	For any leaf $\ell$ in $\form$ labelled by a variable, say $x_{i,j}$, suppose $\mathcal{P}$ is the path from $\ell$ to the root. 
	Consider the set of multiplication gates on $\mathcal{P}$ whose left child is part of $\mathcal{P}$, and let $v$ be the one that is closest to $\ell$.
	Since $\form$ is $\abcd$, the right child of $v$ must be associated with a set, say $[a,b)$.
	If $j \neq a$, we set the label of $\ell$ to zero; otherwise we let it be $x_{i,j}$.
	
	Note that this operation does not kill any \emph{valid} monomial.
	Let $\form'$ be the formula we get by performing the above operation on every leaf of $\form$ that is labelled by a variable.
	$\form'$ is clearly 
	homogeneous and $\abcd$.
	We show that $\form'$ is also \emph{linked}.
	
	Suppose that is not the case.
	Then there is must be a \emph{problematic} vertex in $\form'$.
	Let $v$ be such a vertex of minimal height.
	That is, there is a monomial in the polynomial computed at $v$ in which, say, $x_{i,j}$ appears right before $x_{i',j'}$ but $j \neq i'$.
	Further, the sub-formulas corresponding to the children of $v$ are linked.
	Note that $v$ must be a multiplication gate; not a leaf or an addition gate.

	Let $f_\l$ and $f_\r$ be the polynomials computed at the left and right children of $v$ respectively.
	Also, let $[a,b)$ be the set associated with the right child of $v$.
	Then, it must be the case that the first variable in any monomial in $f_\r$ looks like $x_{a,j'}$ for some $j'$.
	Further, there must be a monomial in $f_\l$ in which the last variable looks like $x_{i,j}$ for $j \neq a$.
	
	Look at the leaf corresponding to this variable.
	Let this leaf be $\ell$ and let $\mathcal{P}$ be the path from $\ell$ to the root.
	Since $x_{i,j}$ is the right most variable in $f_\l$, it must be the case that $v$ is the multiplication gate that is closest to $\ell$, whose left child is on $\mathcal{P}$.
	But then, we should have set $x_{i,j}$ to zero since $j \neq a$.
	Hence, such a monomial can not appear in $f_\l$.

	This shows that $\form'$ is indeed a homogeneous \emph{linked} $\abcd$ formula of size at most that of $\form$ that computes $\lchsym{n}{d}(\vecx)$.
\end{proof}

\noindent The next observation shows that there is a $\poly$-sized homogeneous $\abcd$ formula that computes $\chsym{n}{\log n}(\vecx)$ .

\begin{observation}\label{obs:UBonCHSYMlogDeg}
	$\chsym{n/2}{\log n}(\vecx)$ can be computed by a homogeneous $\abcd$ formula of size $\poly(n)$.
\end{observation}

\begin{proof}
	Consider the following polynomial over variables $\set{t, x_1, \ldots, x_n}$, where we think of $t$ as a commuting variable and $x_1, \ldots, x_n$ as non-commuting variables. 
	\[
		f_{n, d}(\vecx) = \prod_{i=1}^{n} \inparen{1 + \sum_{j=1}^{d} t^j \cdot x_i^j}
	\]
	Note that the coefficient of $t^d$ in $f_{n, d}(\vecx)$ is exactly $\chsym{n}{d}(\vecx)$.
	Further, it is not hard to see that $f_{n/2, \log n}(\vecx)$ is $\abcd$ in terms of $\vecx$ with respect to the bucketing system $\setdef{X_i}{X_i = \set{x_i}}$, and that the given expression results in an $\abcd$ formula of size $O(n (\log n)^2)$. 
	
	Since $t$ is a commuting variable, we can use the usual interpolation techniques \cite{BC92}, to get an $\abcd$ formula computing $\chsym{n/2}{\log n}(\vecx)$ of size $O(n \log n \cdot n (\log n)^2) = O(n^2 (\log n)^3) = \poly(n)$. 
	Since the degree of $\chsym{n/2}{\log n}(\vecx)$ is $O(\log n)$, by \autoref{lem:homogenisation}, there is a homogeneous $\abcd$ formula computing $\chsym{n/2}{\log n}(\vecx)$ of size $\poly(n)$.
\end{proof}

\noindent Another simple observation is that if we are given a homogeneous $\abcd$ formula for an $\abcd$ polynomial, then we almost immediately have one for its various sub-polynomials.

\begin{observation}\label{obs:subformula}
	Suppose there is a homogeneous $\abcd$ formula $\form$ computing a polynomial $f$ that is $\abcd$ with respect to a bucketing system of size $m$. 
	Then, for any $a, b \in [m+1]$, there is a homogeneous $\abcd$ formula $\form_{a,b}$ of size $s$ that computes $f[a,b)$.
\end{observation}

\begin{proof}
	Recall that if $\form$ is a homogeneous $\abcd$ formula computing $f$, then $\form$ is in fact a set of formulas $\setdef{\form_i}{\form_i \text{ computes } f[i,m+1)}$. 
	Consider the formula $\form_a$ and set all variables that belong to buckets $\set{X_b, \ldots, X_m}$ to zero in $\form_a$. 
	This operation clearly kills exactly the monomials in $f[a,m+1)$ that are not in $f[a,b)$.
	Thus if we call this new formula $\form_{a,b}$, then $\form_{a,b}$ is homogeneous, $\abcd$ and computes $f[a,b)$.
\end{proof}

\noindent The next observation is extremely crucial, since it allows us to \emph{amplify the degree} of $\chsym{n}{d}$. 

\begin{observation}\label{obs:amplifyDeg}
	Suppose there is a homogeneous $\abcd$ formula computing $\chsym{n}{d}(\vecx)$ of size $s$, and a homogeneous \emph{linked} $\abcd$ formula computing $\lchsym{n}{d'}(\vecx)$ of size $s'$. Then, there is a homogeneous $\abcd$ formula computing $\chsym{n}{(d \cdot d')}(\vecx)$ of size $(s \cdot s')$.
\end{observation}

\begin{proof}
	Let $\form$ be the homogeneous $\abcd$ formula computing $\chsym{n}{d}(\vecx)$ of size $s$, and $\form'$ be the homogeneous \emph{linked} $\abcd$ formula computing $\lchsym{n}{d'}(\vecx)$ of size $s'$.

	We think of the variable $x_{a,b}$ in $\lchsym{n}{d'}(\vecx)$ as a placeholder for the sub-polynomial $\chsym{n}{d}[a, b+1)(\vecx)$\footnote{Sum of monomials in $\chsym{n}{d}(\vecx)$ whose first variable is $a$ and last variable is one of $\set{x_a, \ldots, x_b}$.} of $\chsym{n}{d}(\vecx)$.
	Note that there is a bijection between monomials in $\chsym{n}{(d \cdot d')}(\vecx)$ and those in the polynomial we get by substituting $x_{a,b}$ in $\lchsym{n}{d'}(\vecx)$ with $\chsym{n}{d}[a, b+1)(\vecx)$.

	By \autoref{obs:subformula}, there is homogeneous $\abcd$ formula $\form_{a,b}$, of size $O(s)$ computing $\chsym{n}{d}[a, b+1)(\vecx)$ for every $a, b \in [n+1]$.
	Thus, if we replace every leaf of $\form'$ labelled by $x_{a,b}$ with $\form_{a,b}$, then the resulting formula is a homogeneous $\abcd$ formula computing $\chsym{n}{(d\cdot d')}(\vecx)$ of size $(s \cdot s')$.
\end{proof}

Finally, we observe that if we are given a homogeneous $\abcd$ formula computing the polynomial $\chsym{(n-d+1)}{d}(\vecx)$, then we get a homogeneous multilinear formula computing the non-commutative version of $\esym{n}{d}(\vecx)$.

\begin{observation}\label{obs:chsymToEsym}
	Consider the \emph{elementary symmetric polynomial}
	\[
		\esym{n}{d}(\vecx) = \sum_{1 \leq i_1 < \ldots < i_d \leq n} x_{i_1} \cdots x_{i_d}.
	\]
	If there is a homogeneous $\abcd$ formula computing $\chsym{(n-d+1)}{d}(\vecx)$ of size $s$, then there is a homogeneous multilinear formula computing $\esym{n}{d}(\vecx)$ of size $s$.
\end{observation}

\begin{proof}
	Suppose $\form$ is a homogeneous $\abcd$ formula computing $\chsym{(n-d+1)}{d}(\vecx)$ of size $s$.
	Since $\form$ is homogeneous, every leaf labelled by a variable can be associated with a position index.
	If a leaf labelled $x_i$ has position $k$ associated with it, then replace the label of that leaf with $x_{i+k-1}$.
	
	Call this formula $\form'$.
	Clearly $\form'$ is a homogeneous formula of size $s$ computing $\esym{n}{d}(\vecx)$.
	Further note that since $\form$ was $\abcd$, $\form'$ is multilinear. 
\end{proof}

\subsection{Proof of the Separation}

We now prove \autoref{thm:main}. 
Let us first recall the statement.

\mainThm*

That $\lchsym{n}{d}(\vecx)$ has a \emph{small} $\abcd$ ABP is not very hard to see.
For the lower bound, we assume that we have been given an $\abcd$ formula $\form$, computing the polynomial $\lchsym{n}{\log n}(\vecx)$, of size $\poly(n)$.
We then keep making changes to this formula till we get a homogeneous multilinear formula computing $\esym{n}{n/2}(\vecx)$ of size $\poly(n)$.
Finally, we use the following theorem of {\Hrubes} and Yehudayoff \cite{HY11} to get a contradiction.

\begin{theorem}[Theorem 1, \cite{HY11}]\label{thm:hom-ml-lb}
	Any homogeneous multilinear formula that computes $\esym{n}{d}(\vecx)$, for $d \leq n/2$, must have size $n \times d^{\Omega(\log d)}$.
\end{theorem}

\noindent Let us now complete the proof of our main theorem.

\begin{proofof}{\autoref{thm:main}}
	An $\abcd$ ABP of size $O(nd)$ computing $\lchsym{n}{d}(\vecx)$ is the following.

	
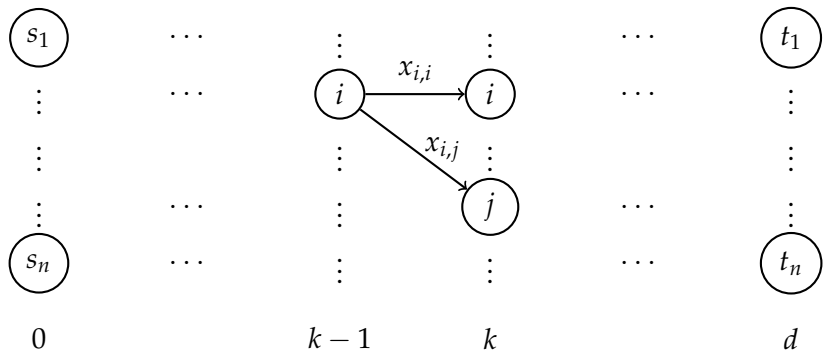
\vspace{1em}
	\begin{tikzpicture}[thick, node distance=3cm]
		\node (start) {};
		\node[circle, draw=black] (l11) at ($(start)+(2,1.5)$) {$s_1$};
		\node (l12) at ($(start)+(2,.75)$) {$\vdots$};
		\node (l10) at ($(start)+(2,0)$) {$\vdots$};
		\node (l13) at ($(start)+(2,-.75)$) {$\vdots$};
		\node[circle, draw=black] (l14) at ($(start)+(2,-1.5)$) {$s_n$};
		\node at ($(start)+(2,-2.5)$) {$0$};

		\node (l21) at ($(start)+(4,1.5)$) {$\cdots$};
		\node (l21) at ($(start)+(4,.75)$) {$\cdots$};
		\node (l23) at ($(start)+(4,-.75)$) {$\cdots$};
		\node (l24) at ($(start)+(4,-1.5)$) {$\cdots$};

		\node (l31) at ($(start)+(6,1.5)$) {$\vdots$};
		\node[circle, draw=black] (l32) at ($(start)+(6,.75)$) {$i$};
		\node (l30) at ($(start)+(6,0)$) {$\vdots$};
		\node (l33) at ($(start)+(6,-.75)$) {$\vdots$};
		\node (l34) at ($(start)+(6,-1.5)$) {$\vdots$};
		\node at ($(start)+(6,-2.5)$) {$k-1$};

		\node (l41) at ($(start)+(8,1.5)$) {$\vdots$};
		\node[circle, draw=black] (l42) at ($(start)+(8,.75)$) {$i$};
		\node (l40) at ($(start)+(8,0)$) {$\vdots$};
		\node[circle, draw=black] (l43) at ($(start)+(8,-.75)$) {$j$};
		\node (l44) at ($(start)+(8,-1.5)$) {$\vdots$};
		\node at ($(start)+(8,-2.5)$) {$k$};

		\node (l51) at ($(start)+(10,1.5)$) {$\cdots$};
		\node (l52) at ($(start)+(10,.75)$) {$\cdots$};
		\node (l53) at ($(start)+(10,-.75)$) {$\cdots$};
		\node (l54) at ($(start)+(10,-1.5)$) {$\cdots$};

		\node[circle, draw=black] (l61) at ($(start)+(12,1.5)$) {$t_1$};
		\node (l62) at ($(start)+(12,.75)$) {$\vdots$};
		\node (l60) at ($(start)+(12,0)$) {$\vdots$};
		\node (l63) at ($(start)+(12,-.75)$) {$\vdots$};
		\node[circle, draw=black] (l64) at ($(start)+(12,-1.5)$) {$t_n$};
		\node at ($(start)+(12,-2.5)$) {$d$};

		\draw[->](l32) -- node[above] {$x_{i,i}$} (l42);
		\draw[->](l32) -- node[right] {$x_{i,j}$} (l43);
	\end{tikzpicture}

	\vspace{1em}
	The ABP has $d+1$ layers, labelled $0$ through $d$, each with $n$ nodes.
	Between any consecutive layers $k-1$ and $k$, where $1 \leq k \leq d$, there is an edge from the $i$-th node in layer $k-1$ to the $j$-th node in layer $k$ layer if $i \leq j$.
	The label on this edge is $x_{i,j}$.
	All the nodes in the first layer are start nodes, and all the ones in the last layer are terminal nodes.

	It is easy to check, by induction, that the polynomial computed between $s_a$ and the $b$-th vertex in layer $k$ computes $\chsym{n}{k}[a,b+1)(\vecx)$. 
	Thus the polynomial computed by the $\abcd$ ABP constructed above is indeed $\chsym{n}{d}(\vecx)$, and its size is clearly $O(nd)$.
	
	Let us now move on to proving the lower bound against $\abcd$ formulas.
	We show that there is a fixed constant $\epsilon_0$ such that any $\abcd$ formula computing $\lchsym{n/2}{\log n}(\vecx)$ must have size atleast $\Omega(n^{\epsilon_0 \log \log n})$.
	Suppose this is not the case. 
	Then for every $\epsilon > 0$, there is an $\abcd$ formula $\form'(\epsilon)$ computing $\lchsym{n/2}{\log n}(\vecx)$ of size $O(n^{\epsilon \log \log n})$.
	
	Without loss of generality, we can assume that $\form'(\epsilon)$ has fan-in $2$.
	Further, by \autoref{lem:depth-reduction}, we can reduce the depth of $\form'(\epsilon)$ to $\log$-depth.
	That is, we get an $\abcd$ formula $\form'_1(\epsilon)$ computing $\lchsym{n/2}{\log n}(\vecx)$ of depth $O(\epsilon \log n \log \log n)$ and size $O(n^{c_1 \epsilon \log \log n})$.
	Here $c_1$ is a fixed constant independent of $\epsilon$. 

	Next, since the degree of the polynomial being computed is \emph{small}, \autoref{lem:homogenisation} implies that $\form'_1(\epsilon)$ can in fact be homogenised without much blow-up in size. 
	In other words, there is a homogeneous $\abcd$ formula computing $\lchsym{n/2}{\log n}(\vecx)$ of size $O(n^{c_1 c_2 \epsilon \log \log n})$, where $c_2$ is again a fixed constant independent of $\epsilon$. 
	Let this formula be $\form'_2(\epsilon)$.

	By \autoref{obs:orderInterval}, we can then use $\form'_2(\epsilon)$ to get a homogeneous linked $\abcd$ formula $\form'_3(\epsilon)$ of size $O(n^{c_1 c_2 \epsilon \log \log n})$ that computes the same polynomial.
	Further, because of \autoref{obs:UBonCHSYMlogDeg}, we know that there is a homogeneous $\abcd$ formula, say $\form$, of size $\poly(n) = O(n^{c_1 c_2 \epsilon \log \log n})$ that computes $\chsym{n/2}{\log n}(\vecx)$.
	
	With $\form$ and $\form'_3(\epsilon)$ in hand, we get a homogeneous $\abcd$ formula $\chsym{n/2}{\log^2 n}(\vecx)$ because of \autoref{obs:amplifyDeg}.
	To get such a formula for $\chsym{n/2}{n/2}(\vecx)$, we need to use \autoref{obs:amplifyDeg} at most $k$ times where 
	\[
		(\log n)^k = \frac{n}{2} \implies k = O \inparen{\frac{\log n}{\log \log n}}.
	\]

	Thus, using \autoref{obs:amplifyDeg} repeatedly at most $O(\sfrac{\log n}{\log \log n})$ times, we get that there is a homogeneous $\abcd$ formula, $\form(\epsilon)$, computing $\chsym{n/2}{n/2}(\vecx)$ of size
	\[
		O(n^{(c_1 c_2 \epsilon \log \log n) \cdot (\sfrac{\log n}{\log \log n})}) = O(n^{(c_1 c_2 \epsilon \log n)}).
	\]
	
	By \autoref{obs:chsymToEsym}, we know that $\form(\epsilon)$ can be used to get a homogeneous multilinear formula, $\form_1(\epsilon)$, computing $\esym{n-1}{n/2}(\vecx)$ of size $O(n^{(c_1 c_2 \epsilon \log n)})$.
	Finally, \autoref{thm:hom-ml-lb} tells us that there is a constant $\delta$ such that any homogeneous multilinear formula computing $\esym{n-1}{n/2}(\vecx)$ must have size at least $n^{\delta \cdot \log n}$.	
	For $\epsilon = \sfrac{\delta}{2 c_1 c_2}$, this contradicts the existence of $\form_1(\epsilon)$ and hence $\form'(\epsilon)$.

	Thus, it must be the case that any $\abcd$ formula computing $\lchsym{n/2}{\log n}(\vecx)$ has size at least $n^{\Omega(\log \log n)}$.
	This completes the proof.
\end{proofof}

\section{Proofs of the Remaining Statements}\label{sec:simpleStatements}

In this section we give proof ideas of the remaining statements mentioned in the introduction.

\subsection{Formula Lower Bounds from Structured Formula Lower Bounds}\label{sec:FormLBfromOrdFormLB}

\FormLBfromOrdFormLB*

\begin{proof}
	By \autoref{thm:main}, we know that the ABP complexity of $\lchsym{\log n}{n}(\vecx)$ is $\poly(n)$.
	Therefore $s \geq n^{\Omega(1)}$.
	Further, note that the polynomial is $\abcd$ with respect to a bucketing system of size $O(\log n)$.
	Therefore, by \autoref{thm:orderForm}, if there is a formula $\form$ computing $\lchsym{\log n}{n}(\vecx)$ of size $s$, then there is an $\abcd$ formula computing it of size $\poly(s)$.
	This immediately implies the given statement.
\end{proof}

\FormLBfromHomFormLB*

\begin{proof}
	By \autoref{thm:main}, we know that the ABP complexity of $\lchsym{n}{\log n}$ is $\poly(n)$.
	Further, the degree of the polynomial is $O(\log n)$.
	Thus, by \autoref{lem:homogenisation}, if there is a formula computing $\lchsym{n}{\log n}(\vecx)$ of size $s$, then there is a homogeneous formula computing it of size $\poly(s)$.
	This immediately implies the given statement.
\end{proof}

\FormLBfromHomDetLB*

\begin{proof}
	Nisan \cite{N91} had shown that the ABP complexity of $\Det{n}(\vecx)$ is $2^{O(n)}$. 
	Thus, by \autoref{lem:homogenisation}, if there is a formula $\form$ computing $\Det{n}(\vecx)$ of size $2^{O(n)}$, then there is a homogeneous formula computing it of size $2^{O(n)}$.
	This immediately implies the given statement.
\end{proof}

\FormLBfromHomImmLB*

\begin{proof}
	Clearly, the ABP complexity of $\imm{n}{\log n}(\vecx)$ is $\poly(n)$.
	Thus, by \autoref{lem:homogenisation}, if there is a formula computing $\imm{n}{\log n}(\vecx)$ of size $s$, then there is a homogeneous formula computing it of size $\poly(s)$.
	This immediately implies the given statement.
\end{proof}

\subsection{Known Relations in the Non-Commutative Setting that Continue to Hold with the Abecedarian Restriction}

\usualInclusions*

\begin{proof}
	Suppose $f \in \abcncVF$. 
	Then $f$ is $\abcd$, and in particular $f \in \ncVF$.
	But we know that $\ncVF \subseteq \ncVBP$, and so $f \in \ncVBP$.
	By \autoref{obs:orderABP}, this implies that $f \in \abcncVBP$.

	Similarly, suppose $f \in \abcncVBP$. 
	Then $f$ is $\abcd$, and $f \in \ncVBP$.
	But $\ncVBP \subseteq \ncVP$, and so $f \in \ncVP$.
	By \autoref{obs:orderCkt}, this implies that $f \in \abcncVP$.
\end{proof}

\ABPtoForm*

\begin{proof}
	The formula we get using the usual divide-and-conquer algorithm has the property that polynomials computed at any of its gate is a polynomial computed between two vertices in the ABP.
	Thus by definition of $\abcd$ ABPs, the statement follows via the usual algorithm.
\end{proof}


\section*{Acknowledgements}
We are thankful to Ramprasad Saptharishi, Mrinal Kumar, C. Ramya and especially Anamay Tengse for the discussions at various stages of this work.
We would also like to thank Ramprasad Saptharishi, Anamay Tengse and Kshitij Gajjar for helping with the presentation of the paper.


\bibliographystyle{alphaurlpp}
\bibliography{references,crossref}
\end{document}